%% file: main.tex
\documentclass[11pt,a4paper]{article}
\usepackage{fullpage}

\newenvironment{frontmatter}{}{}
\newenvironment{keyword}{{\noindent \bf Keywords: }}{}

\input{setup}

\begin{document}

\begin{frontmatter}
  
\title{Asynchronous Byzantine Reliable Broadcast\\
With a Message Adversary}



\author{Timoth\'e Albouy,
           ~Davide Frey,
           ~Michel Raynal,
          ~Fran\c{c}ois Ta\"{i}ani~\\~\\
  Univ Rennes, IRISA, CNRS, Inria, 35042 Rennes, France\\
  \textsf{\small\{timothe.albouy,michel.raynal,francois.taiani\}@irisa.fr, davide.frey@inria.fr}
}

\date{}

\maketitle

\begin{abstract}
This paper considers the problem of reliable broadcast in asynchronous authenticated systems, in which $n$ processes 
communicate using signed messages and 
up to \tb processes may behave arbitrarily (Byzantine processes).
In addition, for each message $m$ broadcast by 
a correct (i.e., non-Byzantine) process, a message adversary may prevent up to \tm correct processes from receiving $m$.
(This message adversary captures network failures
such as transient disconnections, silent churn, or message losses.)
Considering such a ``double'' adversarial context
and assuming $n>3\tb+2\tm$, a reliable broadcast algorithm is presented.
Interestingly,  when there is no message adversary  (i.e., $\tm=0$),
the algorithm terminates in two communication steps
(so, in this case, this algorithm is optimal in terms of both 
Byzantine tolerance and time efficiency).
It is then shown that the condition $n>3\tb+2\tm$ is necessary for implementing reliable broadcast in the presence of both Byzantine processes and a message adversary
(whether the underlying system is enriched with signatures or not).
\end{abstract}







\begin{keyword}
Asynchronous system,
Byzantine processes,
Churn,
Message adversary,
Message losses, 
Message-passing,
Message signatures, 
Reliable broadcast,
Transient disconnection.
\end{keyword}

\end{frontmatter}


\section{Introduction}

\paragraph{Reliable broadcast}
Introduced in the mid-eighties, {\it Reliable Broadcast}
is a fundamental communication abstraction that lies
at the center of fault-tolerant asynchronous distributed systems.
Formally defined in~\cite{B87,BT85}, it allows each process
to broadcast messages in the presence of process failures,
with well-defined delivery properties\footnote{The term {\it delivery}
refers here to the application layer where a process receives
and processes the content of an application message
(see Section~\ref{sec:model}).}.
In turn, these properties make it possible to design provably correct distributed software for 
upper-layer applications based on such a broadcast abstraction.

Intuitively, reliable broadcast guarantees that
the non-faulty processes deliver the same set of messages,
which includes at least all the messages they broadcast.
This set may also contain messages broadcast by faulty processes.
The fundamental property of reliable broadcast lies in the
fact that no two non-faulty processes deliver different sets of
messages~\cite{CGR11,R18}.

When some processes may suffer from Byzantine
failures~\cite{LSP82}, designing a reliable broadcast
communication abstraction that tolerate such failures is far from trivial.
Such an algorithm is called Byzantine-tolerant reliable broadcast (BRB)
and we say that a process \br-broadcasts and \br-delivers messages.
The most famous BRB algorithm is due to Bracha~\cite{B87} (1987).
For an application message, this algorithm
gives rise to three sequential communication steps and up to
$(n-1)(2n+1)$ implementation messages sent by correct processes.
This algorithm requires $n>3\tb$, which is optimal in terms of
fault tolerance.

\paragraph{Recent works related to reliable broadcast}
Due to its fundamental nature, BRB has been addressed by many authors.  
Here are a few recent results.
Similarly to Bracha's algorithm, all these algorithms assume an underlying fully
connected reliable network. 
\begin{itemize}
    \item The versatility dimension of Bracha's algorithm has been
    analyzed in~\cite{HKL20,R21}.
    
    \item Addressing efficiency issues, 
    the BRB algorithm presented in~\cite{IR16} implements 
    the reliable broadcast of an application
    message with only two communication steps and up to
    $n^2-1$ implementation messages sent by correct processes.
    The price to pay for this gain in efficiency
    is a weaker \tb-resilience, namely $\tb < n/5$.
    Hence, this algorithm and Bracha's algorithm differ in their trade-off
    between \tb-resilience and message/time efficiency.
    
    \item Scalable BRB is addressed in~\cite{GKMPS20}.
    The goal of this work is to avoid paying the $O(n^2)$ message
    complexity price.
    To this end, the authors use a non-trivial
    message-gossiping approach which allows them to design a sophisticated
    BRB algorithm satisfying probability-dependent properties.
    
    
    \item
    BRB in dynamic systems is addressed in~\cite{GKKPST20}
    ({\it dynamic} means that a process can enter and leave the system
    at any time).
    In their article, the authors present an efficient BRB algorithm for such a context.
    This algorithm assumes that, at any time, there are at least
    two times more correct processes than Byzantine ones in the system.
    
    \item An efficient algorithm for BRB with long inputs of $b$ bits
    using lower costs than $b$ single-bit instances is presented in~\cite{NRSVX20}.
    This algorithm, which assumes $\tb < n/3$,
    achieves the best possible communication complexity of
    $\Theta(n b)$ input sizes.
    This article also presents an authenticated extension of this solution. 
\end{itemize}

The work presented in this paper\footnote{A very preliminary version of this work appeared in~\cite{AFRT21}.}
goes beyond the previous  proposals by considering the conjunction of two types of adversary: as in the above works, processes may be \textit{Byzantine}, but in addition a \textit{message adversary} may also remove implementation messages between correct processes. 
More precisely, this work addresses the problem of fault-tolerant reliable broadcast 
in asynchronous $n$-process message-passing systems 
enriched with message signatures, 
in which up to \tb processes are Byzantine, 
and a message adversary that may prevent up to \tm
non-Byzantine processes from delivering an implementation message
broadcast by a non-Byzantine process.
This dual fault model originated from our
research on the reconciliation of local process states in
distributed Byzantine-tolerant money transfer systems (a.k.a. cryptocurrencies), in which processes become temporarily disconnected.
Several researchers have indeed pointed out the fundamental role that broadcast abstractions play in Byzantine money transfer systems (see, for instance, \cite{AFRT20,CK21,CGKKMPPSTX20,DKSS22,GKMPS19,GKKPST20}).
This crucial role naturally leads to considering how Byzantine broadcast can be expanded to more volatile and dynamic settings, thus motivating our proposal to combine traditional Byzantine faults with a message adversary.

\begin{table}[ht]
\begin{center}
\renewcommand{\baselinestretch}{1}
\small
\begin{tabular}{|c|c|}
\hline
{\bf Acronyms} & {\bf Meaning}\\
\hline
\hline
BRB & Byzantine-tolerant reliable broadcast \\
\hline
MA & Message adversary \\
\hline
\MBRB & Message adversary- and Byzantine-tolerant reliable broadcast\\
\hline\hline
{\bf Notations} & {\bf Meaning}\\
\hline $n$ & number of processes in the network \\
\hline \tb & upper bound on the number of Byzantine processes \\
\hline \tm & power of the message adversary \\
\hline $c$ & effective number of correct processes in a run ($n-\tb \leq c \leq n$) \\
\hline \lgd & minimal nb of correct processes that \mbr-deliver a message\\
\hline \rtc & time complexity of \MBRB \\
\hline \omc & message complexity of \MBRB \\
\hline
\end{tabular}
\end{center}
\vspace{-0.4cm}
\caption{Acronyms and notations}
\label{terminology}
\end{table}

The paper is made up of~\ref{sec:conclusion} sections.
\begin{itemize}
    \item Section~\ref{sec:model-mbrb} defines the computing model and the Message Adversary-Tolerant Byzantine Reliable Broadcast communication abstraction (or \MBRB for short).
    
    \item Section~\ref{sec:mbrb-impl} presents a signature-based
    algorithm implementing the \MBRB abstraction,
    proves it is correct, and   evaluates its cost.
    When there is no message adversary, this algorithm is optimal 
    from both Byzantine resilience and the number of
    communication steps\footnote{The signature-free BRB algorithm
    described in~\cite{B87} is optimal with respect to
    Byzantine resilience ($\tb<n/3$), but requires three
    communication steps, while the signature-free BRB algorithm
    described in~\cite{IR16} is optimal with respect to
    the number of communication steps (2) but is not with respect
    to Byzantine resilience (it requires $\tb<n/5$).}.
    
    \item Section~\ref{sec:mbrb-tight} shows that the
    condition $n > 3\tb+2\tm$ is necessary and sufficient
    for implementing the \MBRB communication abstraction
    (be the underlying system enriched with signatures or not).
    
    \item Finally, Section~\ref{sec:conclusion} concludes the article.
\end{itemize}

%

\section{Computing Model and MBRB Abstraction}
\label{sec:model-mbrb}

\subsection{Computing Model}
\label{sec:model}

\paragraph{Process model}
The system is composed of $n$ asynchronous sequential processes denoted $p_1$, ..., $p_n$.
Each process $p_i$ has an identity, and all the identities are different and known by all processes.
To simplify, we assume that $i$ is the identity of $p_i$.

Regarding failures, up to \tb processes can be Byzantine, where a Byzantine process is a process whose behavior does not follow the code specified by its algorithm~\cite{LSP82,PSL80}. 
Let us notice that Byzantine processes can collude to fool the non-Byzantine processes (also called correct processes).
Let us also notice that, in this model, the premature stop (crash) of a process is a Byzantine failure.

Moreover, given an execution, $c$ denotes the number of processes that effectively behave correctly in that execution.
We always have $n-\tb \leq c \leq n$.
While this number remains unknown to correct processes, it is used in the following to analyze and characterize (more precisely than using its worse value $n-\tb$) the guarantees provided by the proposed algorithms.

\paragraph{Communication model}
The processes communicate through a fully connected asynchronous point-to-point communication network.
Although this network is assumed to be reliable---in the sense that it neither corrupts, duplicates, nor creates messages---it may nevertheless lose messages due to the actions of a message adversary (defined below).


Let \msgm be a message type and $v$ the associated value. 
A process can invoke the unreliable operation \broadcast $\msgm(v)$, which is a shorthand for ``{\bf for all} $i\in\{1,\cdots, n\}$ {\bf do} \send $\msgm(v)$ {\bf to} $p_j$ {\bf end for}''.
It is assumed that all the correct processes invoke \broadcast to send messages. 
As we can see, the operation \broadcast $\msgm(v)$ is not reliable.
As an example, if the invoking process crashes during its invocation, an arbitrary subset of processes receive the message implementation message $\msgm(v)$.
Moreover, due to its very nature, a Byzantine process can send messages without using the macro-operation \broadcast.

From a terminology point of view, at the system/network level, we say that messages are {\it broadcast} and {\it received}. 
Moreover, a message generated by the algorithm is said to be an {\it implementation} message (\imp in short), while a message generated by the application layer is said to be an {\it application} message (\app in short).


\paragraph{Message adversary}
The notion of a {\it message adversary} (MA) was implicitly introduced in~\cite{SW89} (under the name {\it transient faults}  and {\it ubiquitous faults}) and then used  (sometimes implicitly) in many works (e.g.,~\cite{AG13,CS09,RS13,SW07,TZKZ20}). 
A short tutorial on message adversaries is presented in~\cite{R15}.

Let \tm be an integer constant such that $0 \leq \tm < c$.
The communication network is under the control of an adversary which eliminates \imps sent by processes, so that these \imps appear as being lost.
More precisely, when a correct process invokes \broadcast $\msgm(v)$, the message adversary is allowed to arbitrarily suppress up to \tm copies of the \imp $\msgm(v)$ that were intended to correct processes.
This means that, despite the fact the sender is correct, up to \tm correct processes may miss the \imp $\msgm(v)$\footnote{
A close but different notion was introduced by Dolev in~\cite{D82} (and explored in subsequent works, such as \cite{BDFRT21}) which considers static $k$-connected networks.
If the adversary selects statically for each correct sender \tm correct processes that do not receive this sender's \imps,
the proposed model includes Dolev's model with $k = n-\tm$.}.

As an example, consider a set $D$ of correct processes, where $1\leq |D|\leq d$, such that during some period of time, the adversary suppresses all the \imps sent to them.
It follows that, during this period of time, this set of processes appears as a set of correct processes that are (unknowingly) input-disconnected from the other correct processes.
Depending on the message adversary, the set $D$ may vary with time. 
Let us notice that $\tm=0$ corresponds to the weakest possible message adversary: it corresponds to a classical static system where some processes are Byzantine but no \imp is lost (the network is fully reliable).

Let us remark that this type of message adversary is stronger, and therefore covers, the more specific case of \emph{silent churn}, in which processes (nodes) may decide to disconnect from the network.
While disconnected, such a process silently pauses its algorithm (a legal behavior in our asynchronous model), and is implicitly moved (by the adversary) to the $D$ adversary-defined set.
Upon coming back, the node resumes its execution, and is removed from $D$ by the adversary.\footnote{So the notion of a message adversary implicitly includes the notion of \imp omission failures.}

Informally, in a silent churn environment, a correct process may miss \imps sent by other processes while it is disconnected from the network.
The adjective ``silent'' in {\it silent churn} expresses the fact that no notification is sent on the network by processes whenever they leave or join the system: there is no explicit ``attendance list'' of connected processes, and processes are given no information on the status (connected/disconnected) of their peers.
In this regard, the silent churn model diverges from the classical approach when designing dynamic distributed systems, in which processes send \imps on the network notifying their connection or disconnection \cite{GKKPST20}.
The silent churn model is a good representation of real-life large-scale peer-to-peer systems, where peers leaving the network typically do so in a completely silent manner (i.e., without warning other peers).

Let us also observe that silent churn allows us to model
input-disconnections due to process mobility.
When a process moves from one location to another, the sender's broadcasting range may not be
large enough to ensure that the moving process remains input-connected.
An even more prosaic example would be one where a user simply turns off their device, or disable its Internet connection, preventing it from receiving or sending any further \imps.
In this context, we consider that the message adversary removes all the incoming \imps from the corresponding process until the device reconnects.

Let us mention that the loss of \imps caused by a message adversary may be addressed using a reliable unicast protocol.
These protocols were originally introduced to provide reliable channels on top of an unreliable network subject to \imp losses. 
The principle is simple: the sender keeps sending idempotent \imps at regular intervals through an unreliable channel until it receives an acknowledgement from the receiver.
This principle notoriously lies at the core of the Transmission Control Protocol (TCP), although with important practical adaptations, as TCP uses timeouts to close a malfunctioning or otherwise idle connection, typically after a few minutes.

But because there is no way to detect that a process has crashed or disconnected in an asynchronous environment, an ideal reliable unicast protocol (i.e. one that keeps on re-transmitting until success) needs to treat disconnected processes the same way as slow processes or as if there were packet losses in the network: the sender will thus potentially send infinitely many \imps to a disconnected receiver.
To overcome this issue, some works leverage causal dependencies to avoid resending old \imps: if an acknowledgement is received by the sender for a given \imp, then it can stop resending the \imps that causally precede this \imp and that have not been acknowledged yet (e.g.~\cite{DKSS22}).
However, this approach still assumes that eventually, every communication channel lets some \imps pass, which is not always the case in our model, where the message adversary can permanently sever up to \tm channels.

\paragraph{Digital signatures}
We assume the availability of an asymmetric cryptosystem
to sign data (in practice, \imps) and verify its authenticity.
We assume signatures are secure, and therefore that the computing power of the adversary is bounded.
Every process in the network has a public/private key pair.
We suppose that the public keys are known to everyone, and that the private keys are kept secret by their owner.
Everyone also knows the mapping between any process' identity $i$ and its public key.
Additionally, we suppose that each process can produce at most one signature per \imp.

The signatures are used to cope with the net effect of the Byzantine processes and the fact that \imps broadcast (sent) by correct processes can be eliminated by the message adversary.
A noteworthy advantage of signatures is that, despite the unauthenticated nature of the point-to-point communication channels, signatures allow correct processes to verify the authenticity of \imps that have not been directly received from their initial sender, but rather relayed through intermediary processes.
Signatures provide us with a {\it network-wide} non-repudiation mechanism: if a Byzantine process issues two conflicting \imps to two different subsets of correct processes, then the correct processes can detect the malicious behavior by disclosing to each other the Byzantine signed \imps.\footnote{
The fact that the algorithm uses signed \imps does
not mean that \MBR-broadcast requires signatures to be implemented, see~\cite{AFRT22}.}

\subsection{Message Adversary-Tolerant Byzantine Reliable Broadcast (\MBRB)}
\label{sec:mbrb}
This paper introduces a new broadcast abstraction we have called \textit{Message Adversary-Tolerant Byzantine Reliable Broadcast} (\MBRB for short.)
The \MBRB communication abstraction is composed of two matching operations,
denoted \mbrbroadcast and \mbrdeliver.
It considers that an identity $(i,\sn)$ (sender identity, sequence number) is associated with each \app,
and assumes that any two \apps \mbr-broadcast by the same correct process have different sequence numbers.
Sequence numbers are one of the most natural ways to design ``multi-shot'' reliable broadcast algorithms,
that is, algorithms where the broadcast operation can be invoked
multiple times with different \apps.

When, at the application level, a process $p_i$ invokes $\mbrbroadcast(m,\sn)$, where $m$ is the \app and \sn the associated sequence number, we say $p_i$ ``\mbr-broadcasts $(m,\sn)$''.
Similarly, when $p_i$ invokes $\mbrdeliver(m,\sn,j)$, where $p_j$ is the sender process, we say $p_i$ ``\mbr-delivers $(m,\sn,j)$''.
We say that the \apps are {\it \mbr-broadcast} and {\it \mbr-delivered}
(while, as said previously, the \imps algorithm are {\it broadcast} and {\it received}).

\paragraph{Correctness specification}
Because of the message adversary, we cannot always guarantee that an \app \mbr-delivered by a correct process is eventually \mbr-delivered by all correct processes.
Hence, in the \MBRB specification, we introduce a variable \lgd which indicates the strength of the global delivery guarantee of the primitive:
if one correct process \mbr-delivers an \app then \lgd correct processes eventually \mbr-deliver this \app.\footnote{If there is no message adversary (i.e., $\tm = 0$), we should have $\lgd = c \geq n-\tb$.}
The \MBR-broadcast abstraction is defined by the following properties:
\begin{itemize}
\item Safety:
  \vspace{0.5cm}
\begin{itemize}
    \vspace{-0.6cm}
    \item \MBR-Validity (no spurious message). 
    If a correct process $p_i$ \mbr-delivers an \app $m$ from 
    a correct process $p_j$ with sequence number \sn, then $p_j$ \mbr-broadcast $m$ with sequence number \sn.
    \vspace{-0.1cm}
    \item \MBR-No-duplication.
    A correct process $p_i$ \mbr-delivers at most one \app 
    $m$ from a process $p_j$ with sequence number \sn.
  
    \vspace{-0.1cm}
    \item \MBR-No-duplicity.
    No two different correct processes \mbr-deliver different 
    \apps from a process $p_i$ with the same sequence number~\sn.
\end{itemize}
  \vspace{-0.3cm}
\item Liveness:
  \vspace{-0.1cm}
\begin{itemize}
    \vspace{-0.1cm}
    \item \MBR-Local-delivery.
    If a correct process $p_i$ \mbr-broadcasts an \app $m$ with sequence number \sn, 
    then at least one correct process $p_j$ eventually \mbr-delivers $m$ from $p_i$ with sequence number \sn.
    
    \vspace{-0.1cm}
    \item \MBR-Global-delivery.
    If a correct process $p_i$ \mbr-delivers an \app $m$ from a process $p_j$ with sequence number \sn,
    then at least \lgd correct processes \mbr-deliver $m$ from
    $p_j$ with sequence number \sn.
\end{itemize}
\end{itemize}

It is implicitly assumed that a correct process does not use the same sequence number twice.
Let us observe that, since at the implementation level the message adversary can always suppress all the \imps sent to a fixed set $D$ of \tm processes, the best-guaranteed value for \lgd is $c-\tm$.
Furthermore, let us notice that the constraint $n>2\tm$ prevents the message adversary from partitioning the system.

\paragraph{Performance metrics}
In addition to the correctness specification, we define two metrics that capture the performance of an algorithm implementing the \MBRB specification:
\rtc and \omc, which respectively denote the communication step complexity and
the \imp complexity of the algorithm.
They are defined as follows:
\begin{itemize}
    \vspace{-0.1cm}
    \item \MBR-Time-cost.
    If a correct process $p_i$ \mbr-broadcasts an \app $m$ with sequence number \sn, 
    then \lgd correct processes \mbr-deliver $m$ from $p_i$ with sequence number \sn in at most \rtc communication steps.
    
    \vspace{-0.1cm}
    \item \MBR-Message-cost.
    The \mbr-broadcast of an \app by a correct process $p_i$
    entails the sending of at most \omc \imps by correct processes.
\end{itemize}

\paragraph{Byzantine Reliable Broadcast (BRB)}
If $\lgd=c$ (obtained when $\tm=0$), the previous 
specification boils down to 
Bracha's seminal specification~\cite{B87},
which defines the Byzantine reliable broadcast (BRB) communication abstraction.
Hence, the BRB abstraction is a sub-case of \MBRB.

\section{A Signature-based Algorithm Implementing the \MBRB Abstraction}
\label{sec:mbrb-impl}
This section presents Algorithm~\ref{algo:sb-mbrb}, which implements the
\MBRB communication abstraction in an asynchronous setting
under the constraint $n > 3\tb+2\tm > 0$.
When considering $\tm=0$, this algorithm provides both \tb-tolerance
optimality (as in~\cite{B87}) and step optimality (as in~\cite{IR16}):
it only assumes $n > 3\tb$,
and guaranteed \mbr-delivery of an \app  in only
two communication steps\footnote{Signature-based BRB in only two communication steps is a known result~\cite{ANRX21}, however, to the best of our knowledge, no existing BRB algorithm tolerates message adversaries as well as ours.}.
It follows that signatures can help save one communication step compared to classical signature-free BRB algorithms that assume $\tb < n/3$.
Algorithm~\ref{algo:sb-mbrb} fulfills the \MBR-Global-delivery property with $\lgd = c-\tm$ under the following assumption:
\begin{itemize}
    \item \mbrbassum: $n > 3\tb+2\tm$.
\end{itemize}

\subsection{Preliminaries} \label{sec:mbrb-impl-prel}

\paragraph{Implementation message types}
The algorithm uses only one \imp type, \bundlem,
that carries the signatures backing a given \app $m$, along with $m$'s content, sequence number, and emitter. \bundlem \imps propagate through the network using controlled flooding.

\paragraph{Local data structures}
Each (correct) process saves locally the valid signatures (i.e., the signed fixed-size digests of a certain data)
that it has received from other processes using \bundlem \imps. Each signature ``endorses'' a certain \app $(m,\sn,j)$.
When certain conditions are met (described below), a process further broadcasts in a \bundlem \imp all signatures it knows for a given triplet $(m,\sn,j)$.
A correct process $p_i$ saves at most one signature for a given triplet $(m,\sn,j)$ per signing process $p_k$.

\paragraph{Time measurement}
For the proofs related to  \MBR-Time-cost (Lemmas~\ref{lem:amt-dlv-2-rnd}-\ref{msg-cost}),
we assume that the duration of local computations
is negligible compared to that of \imp transfer delays, and 
consider them to take zero time units.
As the system is asynchronous, the time is measured 
under the traditional assumption that all the \imps 
have the same transfer delay. 
\subsection{Algorithm} \label{sec:mbrb-impl-algo}
At a high level, Algorithm~\ref{algo:sb-mbrb} works by producing, forwarding, and accumulating \emph{witnesses} of an initial \mbr-broadcast operation, until a large-enough quorum 
is observed by at least one correct process, at which point this quorum is propagated in one final
unreliable \broadcast operation.

Witnesses take the form of signatures for a given triplet $(m,\sn,i)$,
where $m$ is the \app, \sn its associated sequence number
and $i$ the identity of the sender $p_i$ (which also produces a signature for $(m,\sn,i)$).
Signatures serve to ascertain the provenance and authenticity of these propagated \bundlem \imps,
thus providing a key ingredient to tolerate the limited reliability of the underlying network.
They also authenticate the invoker of the \mbrbroadcast operation, and finally,
in the last phase of the algorithm, they allow the propagation of a cryptographic proof
that a quorum has been reached, thereby ensuring that enough correct processes eventually
\mbr-deliver the \app that was \mbr-broadcast.
\begin{algorithm}[ht]
\input{sb-mbrb}
\caption{A signature-based implementation of the \MBRB communication abstraction (code for $p_i$)}
\label{algo:sb-mbrb}
\end{algorithm}

In more detail, when a (correct) process $p_i$ invokes $\mbrbroadcast(m,\sn)$, it builds and signs
the triplet $(m,\sn,i)$ to guarantee its non-repudiation,
and saves locally the resulting signature (line~\ref{SB-MBRB-save-own-sig-init}).
Next, $p_i$ broadcasts the \bundlem \imp containing the signature
that it just produced (line~\ref{SB-MBRB-bcast}).

When a correct process $p_i$ receives a $\bundlem(m,\sn,j,\sigs)$ \imp, it first checks if no \app has already been \mbr-delivered for the given sequence number \sn and sender $p_j$, and if $p_j$ signed the \app (line~\ref{SB-MBRB-cond-vld}).
If this condition is satisfied, $p_i$ saves all the new valid signatures
inside the \sigs set (line~\ref{SB-MBRB-save-sigs}).
Next, $p_i$ creates and saves its own signature for $(m,\sn,j)$
and then broadcasts it in a \bundlem \imp, if it has not already done so previously
(line~\ref{SB-MBRB-cond-fwd}-\ref{SB-MBRB-end-cond-fwd}).
Finally, if $p_i$ has saved a quorum of strictly more than $\frac{n+\tb}{2}$ signatures
for the same triplet $(m,\sn,j)$, it broadcasts a \bundlem \imp containing all
these signatures and \mbr-delivers the triplet
(lines~\ref{SB-MBRB-cond-dlv}-\ref{SB-MBRB-end-cond-dlv}).\footnote{The pseudo-code presented in Algorithm~\ref{algo:sb-mbrb} favors readability, and is therefore not fully optimized. For instance, in some cases, a process might unreliably broadcast exactly the same content at lines~ \ref{SB-MBRB-fwd} and~\ref{SB-MBRB-bcast-quorum}. This could be avoided by either using an appropriate flag, or by tracking and preventing the repeated broadcast of identical \bundlem \imps.}

\paragraph{Remark}
The reader can notice that the system parameters $n$ and \tb appear in the algorithm,
whereas the system parameter \tm does not.
Naturally, they all explicitly appear in the proof. 

\subsection{Algorithm proof} \label{sec:mbrb-impl-proof}
This section proves the correctness and performance properties of \MBRB.

\begin{theorem}
\label{theo:sb-mbr-correctness}
If the \mbrbassum is satisfied, Algorithm~{\em  \ref{algo:sb-mbrb}} implements the \mbr-broadcast of an \app by a correct process with the following guarantees:
\begin{itemize}
    \item $\lgd = c-\tm$ correct processes,
    
    \item $\rtc = \left\{\begin{array}{ll}
    2 & \text{if}~~
        \tm < \frac{c - \lfloor\frac{n+\tb}{2}\rfloor}
        {\lfloor\frac{n+\tb}{2}\rfloor +1} \\
    3 & \text{if}~~
        \tm < c-\sqrt{c \times \frac{n+\tb}{2}} \\
    > 3 & \text{otherwise}
    \end{array}\right\}$ communication steps,
    
    \item $\omc = 2n^2$ \imps.
\end{itemize}
\end{theorem}

\noindent The proof follows from the next lemmas.


\begin{lemma}[\MBR-Validity]
If a correct process $p_i$ \mbr-delivers $m$ from a correct process $p_j$ with sequence number \sn, then $p_j$ has previously \mbr-broadcast $m$ with sequence number \sn.
\end{lemma}

\begin{proof} 
If a correct process $p_i$ \mbr-delivers $(m,\sn,j)$ (where $p_j$ is correct) at line~\ref{SB-MBRB-dlv},
then it has passed the condition at line~\ref{SB-MBRB-cond-vld},
which means that it must have witnessed a valid signature for $(m,\sn,j)$ by $p_j$.
Since signatures are secure, the only way to create this signature is for $p_j$ to execute
the instruction at line~\ref{SB-MBRB-save-own-sig-init},
during the $\mbrbroadcast(m,\sn)$ invocation.
\end{proof}


\begin{lemma}[\MBR-No-duplication]
A correct process $p_i$ \mbr-delivers at most one 
\app from a process $p_j$ with a given sequence number \sn.
\end{lemma}

\begin{proof}
This property derives trivially from the condition at line~\ref{SB-MBRB-cond-vld}.
\end{proof}


\begin{lemma}[\MBR-No-duplicity]
No two different correct processes \mbr-deliver different 
\apps from a process $p_i$ with the same sequence number~\sn.
\end{lemma}

\begin{proof}
Let us consider two correct processes $p_a$ and $p_b$ which respectively \mbr-deliver $(m,\sn,i)$ and $(m',\sn,i)$.
Due to the condition at line~\ref{SB-MBRB-cond-dlv},
$p_a$ and $p_b$ must have saved (and thus received) two sets $Q_a$ and $Q_b$ containing
strictly more than $\frac{n+\tb}{2}$ signatures for $(m,\sn,i)$ and $(m',\sn,i)$, respectively.
We thus have $|Q_a| > \frac{n+\tb}{2}$ and $|Q_b| > \frac{n+\tb}{2}$.

As we have $|A \cap B| = |A|+|B|-|A \cup B| \geq |A|+|B|-n > 2\times\frac{n+\tb}{2}-n = \tb$,
$A$ and $B$ have at least one correct process $p_k$ in common,
which must have signed both $(m,\sn,i)$ and $(m',\sn,i)$.
But before signing $(m,\sn,i)$ at line~\ref{SB-MBRB-save-own-sig-init} or~\ref{SB-MBRB-save-own-sig-fwd},
$p_k$ checks that it did not sign a different \app from the same sender and with the same sequence number,
whether it be implicitly during a $\brbroadcast(m,\sn)$ invocation or at line~\ref{SB-MBRB-cond-fwd}.
Thereby, $m$ is necessarily equal to $m'$.
\end{proof}


\begin{lemma}[\MBR-Local-delivery]
If a correct process $p_i$ \mbr-broadcasts an \app $m$ with
sequence number \sn, then at least one correct process $p_j$
\mbr-delivers $m$ from $p_i$ with sequence number \sn.
\end{lemma}

\begin{proof}
If a correct process $p_i$ \mbr-broadcasts $(m,\sn)$,
then it broadcasts its own signature $\sig_i$ for $(m,\sn,i)$
in a $\bundlem(m,\sn,i,\{\sig_i\})$ message at line~\ref{SB-MBRB-bcast}.
As $p_i$ is correct, it does not sign another triplet $(m',\sn,i)$ where $m' \neq m$,
therefore it is impossible for a correct process to \mbr-deliver $(m',\sn,i)$ at
line~\ref{SB-MBRB-dlv}, because it cannot pass the condition at line~\ref{SB-MBRB-cond-vld}.

Let us denote by $K$ the set of correct processes that receive a message $\bundlem(m,\sn,i,\{\sig_i,...\})$ at least once.
The first one of such \bundlem messages that a process of $K$ receives can be the one $p_i$ initially broadcast at line~\ref{SB-MBRB-bcast}, but it can also be a \bundlem message broadcast by a correct process at lines~\ref{SB-MBRB-fwd} or~\ref{SB-MBRB-bcast-quorum}, or it can even be a \bundlem message sent by a Byzantine process.
In any case, the first time the processes of $K$ receive such a \bundlem message, they pass the condition at line~\ref{SB-MBRB-cond-vld}, and they also pass the condition at line~\ref{SB-MBRB-cond-fwd}, except for $p_i$ if it belongs to $K$.
Consequently, each process $p_k$ of $K$ necessarily broadcasts its own signature $\sig_k$ for $(m,\sn,i)$ in a $\bundlem(m,\sn,i,\{\sig_k,\sig_i,...\})$ message.

\sloppy
By construction of the algorithm, the set $K$ of correct processes that receive a $\bundlem(m,\sn,i,\{\sig_i,...\})$ message is equal to the set of correct processes $p_k$ that broadcast a $\bundlem(m,\sn,i,\{\sig_k,\sig_i,...\})$.
By the definition of the message adversary, a message $\bundlem(m,\sn,i,\{\sig_k,\sig_i,...\})$ broadcast by a correct process $p_k$ is eventually received by at least $c-\tm$ correct processes.
Hence, the minimum number of signatures for $(m,\sn,i)$ made by processes of $K$ that is also received by processes of $K$ globally is $|K|(c-\tm)$.
It follows that a given process of $K$ individually receives on average the distinct signatures of at least $|K|(c-\tm)/|K| = c-\tm$ processes of $K$.

From \mbrbassum, we have $3\tb+2\tm < n \iff n+3\tb+2\tm < 2n \iff n+\tb < 2n-2\tb-2\tm \iff \frac{n+\tb}{2} < n-\tb-\tm \leq c-\tm$ (as $n-\tb \leq c$).
As a result, at least one process $p_j$ of $K$ (ergo one correct process) receives a set $S$ (in possibly multiple \bundlem messages) of strictly more than $\frac{n+\tb}{2}$ valid distinct signatures for $(m,\sn,i)$.
When $p_j$ receives the last signature of $S$, there are two cases:
\begin{itemize}
    \item Case if $p_j$ does not pass the condition at line~\ref{SB-MBRB-cond-vld}.
    
    As processes of $K$ are correct, then when they broadcast a
    $\bundlem(m,\sn,i,\sigs)$ message, they necessarily include $\sig_i$ in \sigs,
    which implies that $\sig_i$ is necessarily in $S$.
    Therefore, if $p_j$ does not pass the condition at line~\ref{SB-MBRB-cond-vld},
    it is because $p_j$ already \mbr-delivered some $(-,\sn,i)$.
    But let us remind that, as $p_i$ is correct, it is impossible for $p_j$
    to \mbr-deliver anything different from $(m,\sn,i)$.
    Therefore, $p_j$ has already \mbr-delivered $(m,\sn,i)$.
    
    \item Case if $p_j$ passes the condition at line~\ref{SB-MBRB-cond-vld}.
    
    Process $p_j$ then saves all signatures of $S$ at line~\ref{SB-MBRB-save-sigs},
    and after it passes the condition at line~\ref{SB-MBRB-cond-dlv} (as
    $|S| > \frac{n+\tb}{2}$) and finally \mbr-delivers
    $(m,\sn,i)$ at line~\ref{SB-MBRB-dlv}. \qedhere
\end{itemize}
\end{proof}


\begin{lemma}[\MBR-Global-delivery]
If a correct process $p_i$ \mbr-delivers an \app $m$ 
from $p_j$ with sequence number \sn, then at least $\lgd = c-\tm$ correct
processes \mbr-deliver $m$ from $p_j$ with sequence number \sn.
\end{lemma}

\begin{proof}
If a correct process $p_i$ \mbr-delivers $(m,\sn,j)$ at line~\ref{SB-MBRB-dlv}, it must have saved a set \sigs of strictly more than $\frac{n+\tb}{2}$ valid distinct signatures because of the condition at line~\ref{SB-MBRB-cond-dlv}.
Let us remark that \sigs necessarily contains the signature for $(m,\sn,i)$ by $p_i$ because of the condition at line~\ref{SB-MBRB-cond-vld}.
Additionally, $p_i$ must also have broadcast  $\bundlem(m,\sn,i,\sigs)$ at line~\ref{SB-MBRB-bcast-quorum}, that, by definition of the message adversary, is received by a set $K$ of at least $c-\tm$ correct processes.
For each process $p_k$ of $K$:
\begin{itemize}
    \item If $p_k$ does not pass the condition at line~\ref{SB-MBRB-cond-vld}, it is necessarily because it has already \mbr-delivered some $(-,\sn,j)$ at line~\ref{SB-MBRB-dlv}.
    But because of \MBR-No-duplicity, $p_k$ has necessarily \mbr-delivered $(m,\sn,j)$.
    
    \item If $p_k$ passes the condition at line~\ref{SB-MBRB-cond-vld}, then it saves all signatures of \sigs at line~\ref{SB-MBRB-save-sigs} and then passes the condition at line~\ref{SB-MBRB-cond-dlv} and finally \mbr-delivers $(m,\sn,j)$ at line~\ref{SB-MBRB-dlv}.
\end{itemize}

Therefore, all processes of $K$ (which, as a reminder, are at least $c-\tm=\lgd$) necessarily \mbr-deliver $(m,\sn,j)$ at line~\ref{SB-MBRB-dlv}.
\end{proof}


\begin{lemma} \label{lem:sufficient-for-quorum}
$c-\tm > \big\lfloor \frac{n+\tb}{2} \big\rfloor$.
\end{lemma}

\begin{proof}
We have the following:
\begin{align*}
    c-\tm &\geq n-\tb-\tm
    = \frac{2n-2\tb-2\tm}{2},\tag{by definition of $c$}\\
    &> \frac{n+3\tb+2\tm-2\tb-2\tm}{2}, \tag{by \mbrbassum}\\
    &> \frac{n+\tb}{2} \geq \Big\lfloor \frac{n+\tb}{2} \Big\rfloor.\qedhere
\end{align*}
\end{proof}

\begin{lemma} \label{lem:amt-dlv-2-rnd}
If a correct process $p_i$ \mbr-broadcasts $(m,\sn)$, then at least $c-\tm-\Big\lfloor\frac{\tm\lfloor\frac{n+\tb}{2}\rfloor}{c-\tm-\lfloor\frac{n+\tb}{2}\rfloor}\Big\rfloor$ correct processes \mbr-deliver $(m,\sn,i)$ at most two communication steps later.
\end{lemma}


\begin{proof}
If a correct process $p_i$ \mbr-broadcasts $(m,\sn)$, then it broadcasts its own signature $\sig_i$ for $(m,\sn,i)$ in a $\bundlem(m,\sn,i,\{\sig_i\})$ \imp at line~\ref{SB-MBRB-bcast}.
Let us denote by $K$ the set of correct processes that receive this $\bundlem(m,\sn,i,\{\sig_i\})$ \imp from $p_i$ during the same communication step, and let \klem be the number of processes in $K$, such that $c-\tm \leq \klem = |K| \leq c$ (by definition of the message adversary).
By construction of the algorithm, every process $p_x$ of $K$ passes the condition at line~\ref{SB-MBRB-cond-vld}, and therefore broadcasts a $\bundlem(m,\sn,i,\{\sig_x,\sig_i\})$ \imp, whether it be at line~\ref{SB-MBRB-bcast} for $p_i$, or at line~\ref{SB-MBRB-fwd} for any other process of $K$.

Let $A$ and $B$ define two partitions of the set of all correct processes
($A \cup B$ is the set of all correct processes, and $A \cap B = \varnothing$).
$A$ denotes the set of correct processes that receive strictly more than
$\frac{n+\tb}{2}$ signatures for $(m,\sn,i)$ from processes of $K$ two communication steps after $p_i$ \mbr-broadcast $(m,\sn)$,
while $B$ denotes the set of remaining correct processes
of $K$ that receive at most $\frac{n+\tb}{2}$ signatures for $(m,\sn,i)$
from processes of $K$ two communication steps
after $p_i$ \mbr-broadcast $(m,\sn)$.
Let \llem be the size of $A$: $\llem = |A|$.
By construction, $|B|=c-\llem$.
Let \sAc and \sBc respectively denote the number of signatures for
$(m,\sn,i)$ from processes of $K$ received by processes of $A$ and $B$ at most two communication steps after $p_i$ \mbr-broadcast $(m,\sn)$.
Figure~\ref{fig:msg-dist} represents the distribution of
such signatures among processes of $K$,
sorted by decreasing number of signatures received.
Each processes of $A$ can receive at most \klem signatures (that is,
all signatures) from processes of $K$, while each process of $B$
can receive at most $\lfloor\frac{n+\tb}{2}\rfloor$ signatures
from processes of $K$ two communication
steps after $p_i$ \mbr-broadcasts $(m,\sn)$.
For the sake of simplicity, we use \qlem in the place of
$\lfloor\frac{n+\tb}{2}\rfloor$ in some parts of this proof.

\begin{figure}[ht]
\input{msg-dist}
\caption{Distribution of signatures among processes of $A$ and $B$ two communication steps after $p_i$ \mbr-broadcast $(m,\sn)$}
\label{fig:msg-dist}
\end{figure}
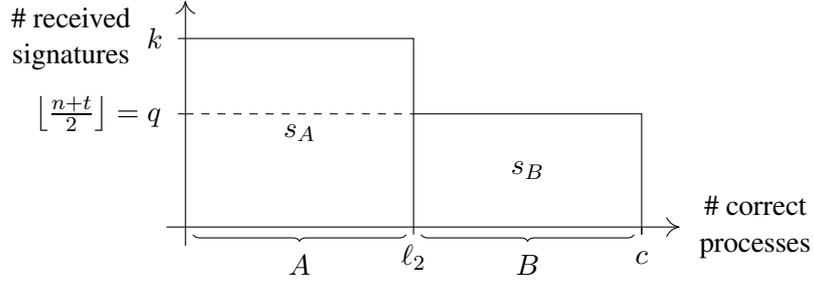

\noindent From these observations, we infer the following inequalities:
\begin{align*}
    \llem\klem &\geq \sAc,\\
    (c-\llem) \qlem &\geq \sBc.
\end{align*}

By the definition of the message adversary, a $\bundlem(m,\sn,i,\{\sig_x,\sig_i\})$ \imp broadcast by a correct process $p_x$ is eventually received by at least $c-\tm$ correct processes.
As a consequence, in total, the minimum number of signatures for $(m,\sn,i)$ collectively received by correct processes as a result of broadcasts by processes in $K$ in the first two asynchronous communication steps is $\klem(c-\tm)$.
We thus have:
\begin{align*}
    \sAc+\sBc &\geq \klem(c-\tm).
\end{align*}
By combining the previous inequalities, we obtain:
\begin{align}
    \llem\klem+(c-\llem)\qlem &\geq \klem(c-\tm), \nonumber\\
    \llem\klem+c\qlem-\llem\qlem &\geq \klem(c-\tm), \nonumber\\
    \llem\klem-\llem\qlem &\geq \klem(c-\tm)-c\qlem, \nonumber\\
    \llem (\klem-\qlem) &\geq \klem(c-\tm)-c\qlem. \label{eq:l-fact-sup-k-fact}
\end{align}
By Lemma~\ref{lem:sufficient-for-quorum},
we know that $\klem \geq c-\tm > \big\lfloor \frac{n+\tb}{2} \big\rfloor = \qlem$,
so we can rewrite~(\ref{eq:l-fact-sup-k-fact}) into:
\begin{align}
    \llem &\geq \frac{\klem(c-\tm)-c\qlem}{\klem-\qlem}. \label{eq:min-l-with-k}
\end{align}

Let us define a function \fnlem such that $\fnlem(\klem) = \frac{\klem(c-\tm)-c\qlem}{\klem-\qlem}$.
As we seek the lowest guaranteed value for \llem, we want to find
the minimum of \fnlem on $\klem \in [c-\tm,c]$.
To this end, let us first study the derivative of \fnlem.
The image $\fnlem(\klem)$ is of the form $\frac{u}{v}$, so we have:
\begin{align*}
    \fnlem'(\klem) &= \frac{u'v-uv'}{v^2}
    = \frac{(c-\tm)(\klem-\qlem)-(\klem(c-\tm)-\qlem c)}{(\klem-\qlem)^2},\\
    &= \frac{(c-\tm)(\klem-\qlem)-\klem(c-\tm)+\qlem c}{(\klem-\qlem)^2}
    = \frac{\qlem c-\qlem(c-\tm)}{(\klem-\qlem)^2}
    = \frac{\qlem\tm}{(\klem-\qlem)^2}.
\end{align*}

As $q$ and $d$ are by definition positive, 
we know that
$\fnlem'(\klem) = \frac{\qlem\tm}{(\klem-\qlem)^2}$ is positive, or null when $\tm=0$.
Therefore, \fnlem is monotonically increasing on $\klem \in [c-\tm,c]$,
and the minimum value for \llem can be found when
\klem is also minimum, that is, when $\klem = c-\tm$.
Thus, when we replace \klem by $c-\tm$
in~(\ref{eq:min-l-with-k}), we obtain:
\begin{align}
    \llem &\geq \frac{(c-\tm)(c-\tm)-c\qlem}{c-\tm-\qlem}
    = \frac{(c-\tm)(c-\tm-\qlem)-\qlem\tm}{c-\tm-\qlem}, \nonumber\\
    &\geq c-\tm-\frac{\qlem\tm}{c-\tm-\qlem}. \label{eq:min-l-without-k}
\end{align}

Let us denote by \llemmin the minimum number of correct processes that receive a quorum
of strictly more than $\frac{n+\tb}{2}$ valid distinct signatures for $(m,\sn,i)$
two communication steps after $p_i$ \mbr-broadcast $(m,\sn)$, such that $\llemmin \leq \llem = |A|$.
As the right hand side of~(\ref{eq:min-l-without-k}) is not always an integer, we have:
\begin{align*}
    \llemmin &= \Big\lceil c-\tm-\frac{\qlem\tm}{c-\tm-\qlem} \Big\rceil
    = c-\tm+\Big\lceil-\frac{\qlem\tm}{c-\tm-\qlem} \Big\rceil,\\
    &= c-\tm-\Big\lfloor \frac{\qlem\tm}{c-\tm-\qlem} \Big\rfloor,
    \tag{as $\Forall x \in \mathbb{R}, \lceil -x \rceil = -\lfloor x \rfloor$}\\
    &= c-\tm-\Big\lfloor\frac{\tm\lfloor\frac{n+\tb}{2}\rfloor}{c-\tm-\lfloor\frac{n+\tb}{2}\rfloor}\Big\rfloor.
    \tag{by definition of \qlem}
\end{align*}

Hence, at least $\llemmin = c-\tm-\Big\lfloor\frac{\tm\lfloor\frac{n+\tb}{2}\rfloor}{c-\tm-\lfloor\frac{n+\tb}{2}\rfloor} \Big\rfloor$
processes of $K$ receive strictly more than $\frac{n+\tb}{2}$ valid distinct signatures
for $(m,\sn,i)$ two communication steps after $p_i$ \mbr-broadcasts $(m,\sn)$.
For every process $p_a$ of $A$:
\begin{itemize}
    \item If $p_a$ does not pass the condition at line~\ref{SB-MBRB-cond-vld} after receiving the last signature of the quorum in a \bundlem \imp, it is necessarily because $p_a$ already \mbr-delivered some $(-,\sn,i)$, because processes of $K$ are correct and all their \bundlem \imps include the signature for $(m,\sn,i)$ by $p_i$.
    But let us remind that, as the sender $p_i$ is correct, it is impossible for $p_a$ to \mbr-deliver anything different from $(m,\sn,i)$.
    Therefore, $p_a$ has already \mbr-delivered $(m,\sn,i)$ at line~\ref{SB-MBRB-dlv}.
    
    \item If $p_a$ passes the condition at line~\ref{SB-MBRB-cond-vld} after processing the last $\bundlem(m,\sn,i,\{\sig_i,\sig_x\})$ \imp of the quorum from a process $p_x$, then $p_a$ saves the signature $\sig_x$ at line~\ref{SB-MBRB-save-sigs}, and after it passes the condition at line~\ref{SB-MBRB-cond-dlv} (as it has saved strictly more than $\frac{n+\tb}{2}$ signatures) and finally \mbr-delivers $(m,\sn,i)$ at line~\ref{SB-MBRB-dlv}.
\end{itemize}

Therefore, all processes of $A$, which are at least
$\llemmin = c-\tm-\Big\lfloor\frac{\tm\lfloor\frac{n+\tb}{2}\rfloor}{c-\tm-\lfloor\frac{n+\tb}{2}\rfloor}\Big\rfloor$,
\mbr-deliver $(m,\sn,i)$ at line~\ref{SB-MBRB-dlv} at most two communication steps after $p_i$ \mbr-broadcast $(m,\sn)$.
\end{proof}


\begin{lemma} \label{lem:dlv-3-rnd-if-cond}
If a correct process $p_i$ \mbr-broadcasts $(m,\sn)$ and $\tm < c-\sqrt{c \times \frac{n+\tb}{2}}$, then at least $c-\tm$ correct processes \mbr-deliver $(m,\sn,i)$ at most three communication steps later.
\end{lemma}

\begin{proof}
Let us assume that a correct process $p_i$ \mbr-broadcasts $(m,\sn)$ and that $\tm < c-\sqrt{c \times \frac{n+\tb}{2}}$.
Process $p_i$ must unreliably broadcast a first $\bundlem(m,\sn,i,\{\sig_i\})$ \imp (where $\sig_i$ is the signature of $(m,\sn,i)$ by $p_i$) at line~\ref{SB-MBRB-bcast}.
This initial \imp is received by at least $(c-\tm-1)$ other correct processes, due to our assumption on the message adversary. This counts for a first communication step.
    
In the second communication step, each process $p_j$ of these $(c-\tm-1)$ correct processes unreliably broadcasts its own $\bundlem(m,\sn,i,\{\sig_j,\sig_i\})$ \imp (where $\sig_j$ is the signature of $(m,\sn,i)$ by $p_j$) at line~\ref{SB-MBRB-fwd}.
At the end of the second communication step, in total, at least $(c-\tm)$ distinct signatures for $(m,\sn,i)$ have been created and unreliably broadcast by correct processes (counting that of $p_i$), resulting in at least $(c-\tm)^2$ receptions of said signatures by correct processes.
As there are $c$ correct processes, this means that, on average, each correct process has received at least $\frac{(c-\tm)^2}{c}$ signatures by the end of the second communication step, and that at least one correct process, $p_k$, receives (and saves at line~\ref{SB-MBRB-save-sigs}) at least this number of signatures.

From the Lemma hypothesis $\tm < c-\sqrt{c \times \frac{n+\tb}{2}}$ and using simple algebraic transformations, we can derive $\frac{(c-\tm)^2}{c} > \frac{n+\tb}{2}$.
Therefore, $p_k$ reaches a quorum of signatures, that is, it passes the condition at line~\ref{SB-MBRB-cond-dlv} and unreliably broadcast this quorum of signatures at line~\ref{SB-MBRB-bcast-quorum}, two communication steps after the \mbr-broadcast of $(m,\sn)$ by $p_i$.
By definition of the message adversary, this quorum of signatures is received by $c-\tm$ correct processes, which save it at line~\ref{SB-MBRB-save-sigs} and thus pass the condition at line~\ref{SB-MBRB-cond-dlv} and finally \mbr-deliver $(m,\sn,i)$ at line~\ref{SB-MBRB-dlv}, three communication steps after the \mbr-broadcast of $(m,\sn)$ by $p_i$.
\end{proof}


\begin{lemma}[\MBR-Time-cost] \label{lem:mbrb-time7-cost}
If a correct process $p_i$ \mbr-broadcasts an \app $m$ with sequence number \sn, then $\lgd = c-\tm$ correct processes \mbr-deliver $m$ from $p_i$ with sequence number \sn at most \\~\\
$\rtc = \left\{\begin{array}{ll}
2 & \text{if}~~
    \tm < \frac{c-\lfloor\frac{n+\tb}{2}\rfloor}
    {\lfloor\frac{n+\tb}{2}\rfloor +1} \\
3 & \text{if}~~
    \tm < c-\sqrt{c \times\frac{n+\tb}{2}}\\
>3 & \text{otherwise}
\end{array}\right\}$ communication steps later.
\end{lemma}     

\begin{proof}
Let us consider a correct process $p_i$ that \mbr-broadcasts $(m,\sn)$.
By exhaustion:
\begin{itemize}
    \item Case where
    $\tm < \frac{c-\lfloor\frac{n+\tb}{2}\rfloor}
        {\lfloor\frac{n+\tb}{2}\rfloor +1}$.
    
    By Lemma~\ref{lem:amt-dlv-2-rnd}, at least $c-\tm-\Big\lfloor\frac{\tm\lfloor\frac{n+\tb}{2}\rfloor}{c-\tm-\lfloor\frac{n+\tb}{2}\rfloor}\Big\rfloor$ correct processes \mbr-deliver $(m,\sn,i)$ two communication steps after $p_i$ has \mbr-broadcast $(m,\sn)$.
    We have:
    \begin{align*}
        \tm &< \frac{c-\lfloor\frac{n+\tb}{2}\rfloor}{\lfloor\frac{n+\tb}{2}\rfloor +1}, \tag{case assumption}\\
        \tm\Big\lfloor\frac{n+\tb}{2}\Big\rfloor+\tm &< c-\Big\lfloor\frac{n+\tb}{2}\Big\rfloor, \tag{as $\lfloor\frac{n+\tb}{2}\rfloor +1 > 0$}\\
        \tm\Big\lfloor\frac{n+\tb}{2}\Big\rfloor &< c-\tm-\Big\lfloor\frac{n+\tb}{2}\Big\rfloor, \\
        \frac{\tm\lfloor\frac{n+\tb}{2}\rfloor}{c-\tm-\lfloor\frac{n+\tb}{2}\rfloor} &< 1, \tag{as $c-\tm > \lfloor\frac{n+\tb}{2}\rfloor$ by Lemma~\ref{lem:sufficient-for-quorum}}\\
        \left\lfloor\frac{\tm\lfloor\frac{n+\tb}{2}\rfloor}{c-\tm-\lfloor\frac{n+\tb}{2}\rfloor}\right\rfloor &\leq 0, \\
        c-\tm-\left\lfloor\frac{\tm\lfloor\frac{n+\tb}{2}\rfloor}{c-\tm-\lfloor\frac{n+\tb}{2}\rfloor}\right\rfloor &\geq c-\tm = \lgd.
    \end{align*}
    Hence, \lgd correct processes \mbr-deliver $(m,\sn,i)$
    at most two communication steps
    after $p_i$ has \mbr-broadcast $(m,\sn)$.
    
    \item Case where $\tm < c-\sqrt{c \times \frac{n+\tb}{2}}$.
    
    Lemma~\ref{lem:dlv-3-rnd-if-cond} applies and at least $c-\tm = \lgd$ correct processes \mbr-deliver $(m,\sn,i)$ at most three communication steps after $p_i$ has \mbr-broadcast $(m,\sn)$.
    \qedhere
\end{itemize}
\end{proof}

\begin{lemma}[\MBR-Message-cost]
\label{msg-cost}
The \mbr-broadcast of an \app by a correct process $p_i$
entails the sending of at most $\omc = 2n^2$ \imps by correct processes.
\end{lemma}     

\begin{proof}
The broadcast of an \imp by a correct process at line~\ref{SB-MBRB-bcast} entails its forwarding by at most $n-1$ other correct processes at line~\ref{SB-MBRB-fwd}.
As each broadcast by correct process corresponds to the sending of $n$ \imps, then at most $n^2$ \imps are sent in a first step.

In a second step, at least one correct process reaches a quorum of signatures and passes the condition at line~\ref{SB-MBRB-cond-dlv}, and then broadcasts this quorum of signatures at line~\ref{SB-MBRB-bcast-quorum}.
Upon receiving this quorum, every correct process also passes the condition at line~\ref{SB-MBRB-cond-dlv} (if it has not done it already) and broadcasts the \imp containing the quorum at line~\ref{SB-MBRB-bcast-quorum}.
Hence, at most $n^2$ \imps are also sent in this second step, which amounts to a maximum of $\omc = 2n^2$ \imps sent in total.
\end{proof}


\paragraph{An additional property}
The reader can check from the previous proofs that the algorithm satisfies the following \MBR-delivery property.
If there is a set $K$ of $k$ correct processes, $1 \leq k \leq \tm$, such that there is a finite time $\tau$ after which the message adversary never eliminates the \imps sent to them, then, after $\tau$, each process of $K$ \mbr-delivers all the \apps \mbr-broadcast by correct processes.


\section{A Tightness Bound} \label{sec:mbrb-tight}

\paragraph{Definition}
An algorithm implementing a broadcast communication abstraction
is {\it event-driven} if, as far as the correct processes are
concerned, only (i) the invocation of the
broadcast operation that is provided to the application by the broadcast communication abstraction, or
(ii) the reception of an \imp---sent by a correct or a Byzantine process---can generate the sending of \imps
(using the underlying unreliable network-level \broadcast operation). 

\begin{theorem}[\MBR-Necessary-condition]
\label{theo-necessity}
When $n \leq 3\tb+2\tm$, there is no event-driven 
(signature-free or signature-based) algorithm
implementing the {\em \MBRB} communication abstraction on top of an
$n$-process asynchronous system in which up to \tb processes may be
Byzantine and where a message adversary may suppress up to
\tm copies of each \imp broadcast by a
correct process.\footnote{Let us recall that the underlying
communication operation offered by the system is an unreliable
broadcast defined in Section~\ref{sec:model}.}
\end{theorem}

\begin{proof}
Without loss of generality the proof considers the case
$n = 3\tb+2\tm$.  
Let us partition the $n$ processes into five sets $Q_1, Q_2, Q_3$, $D_1$,
and $D_2$, such that $|D_1|= |D_2|= \tm$ and
$|Q_1|= |Q_2|=|Q_3| = \tb$.\footnote{For the case $n < 3\tb+2\tm$,
the partition is such that $\mmax(|Q_1|,|D_2|) \leq \tm$ and
$\mmax(|Q_1|,|Q_2|,|Q_3|) \leq \tb$.}
So, when considering the sets $Q_1$, $Q_2$, and  $Q_3$, there are
executions in which all the processes of either $Q_1$ or $Q_2$ or
$Q_3$ can be Byzantine, while the processes of the two other sets are not.

The proof is by contradiction. So,  assuming that there is an
event-driven algorithm $A$ that builds the \MBR-broadcast abstraction
for $n = 3\tb+2\tm$, let us consider an execution $E$ of $A$ in which the
processes of $Q_1$, $Q_2$,  $D_1$, and $Q_2$  are not Byzantine while
all the processes of $Q_3$ are Byzantine.

Let us observe that the message adversary can isolate up to \tm
processes by preventing them from receiving any \imps. 
Without loss of generality, let us assume that the adversary
isolates a set of \tm correct processes not containing the sender of
the \app. As $A$ is event-driven,  these \tm isolated
processes do not send \imps during the execution $E$ 
of $A$. As a result, no correct process can expect \imps
from more than $(n-\tb-\tm)$ different processes without risking
being blocked forever. Thanks to the \mbrbassum $n = 3\tb+2\tm$, this
translates as ``no correct process can expect \imps
from more than $(2\tb+\tm)$ different processes without risking being
blocked forever''.

In the execution $E$, the (Byzantine) processes of $Q_3$ simulate the
\mbr-broadcast of an \app such that this \app appears
as being \mbr-broadcast by one of them and is \mbr-delivered as the
\app $m$ to the processes of $Q_1$ (hence the processes
of $Q_3$ appear, to the processes of $Q_1$, as if they were correct)
and as the \app $m' \neq m$ to the processes of $Q_2$
(hence, similarly to the previous case, the processes of $Q_3$ appear
to the processes of $Q_2$ as if they were correct).
Let us call $m$-messages (resp., $m'$-messages) the \imps generated by the event-driven algorithm $A$ that entails the \mbr-delivery of $m$ (resp., $m'$). Moreover, the execution $E$ is such that:
\begin{itemize}
    \item concerning the $m$-messages: the message adversary suppresses all the $m$-messages sent to the processes of $D_2$, and asynchrony delays the reception of all the $m$-messages sent to $Q_2$ until some time $\tau$ defined below.\footnote{
    Equivalently, we could also say that asynchrony delays the reception of all the $m$-messages sent to $D_2 \cup Q_2$ until time $\tau$.
    The important point is here that, due to the assumed existence of Algorithm $A$, the processes of $Q_1$ and and $D_1$ \mbr-deliver $m$ with $m$-messages from at most $2\tb+\tm$ different processes.}
    So, as $|Q_1 \cup D_1 \cup Q_3| =n-\tb-\tm=2\tb+\tm$, Algorithm A will cause the processes of $Q_1$ and $D_1$ to \mbr-deliver~$m$.\footnote{
    Let us notice that this is independent from the fact that the processes in $Q_3$ are Byzantine or not.}
    
    \item concerning the $m'$-messages:
    the message adversary suppresses all the $m'$-messages sent to the processes of $D_1$, and the asynchrony delays the reception of all the $m'$-messages sent to $Q_1$ until time $\tau$.
    As previously, as $|Q_2 \cup D_2 \cup Q_3| =n-\tb-\tm=2\tb+\tm$, Algorithm $A$ will cause the
    processes of $Q_2$ and $D_2$ to \mbr-deliver~$m'$.

    \item Finally, the time $\tau$ occurs after the \mbr-delivery
    of $m$ by the processes of $D_1$ and $Q_1$, and after
    the \mbr-delivery of $m'$ by the processes of $D_2$ and $Q_2$.
\end{itemize}

It follows that different non-Byzantine processes \mbr-deliver
different \apps for the same \mbr-broadcast
(or a fraudulent simulation of it) issued by a Byzantine process
(with possibly the help of other Byzantine processes).
This contradicts the \MBR-No-Duplicity property, which
concludes the proof of the theorem.
\end{proof}

\begin{theorem}[Algorithm optimality]
\label{theo-nce-and-suff}
Considering an asynchronous $n$-process
system in which up to \tb processes can be Byzantine and where a
\tm-message adversary can suppress \imps,
Algorithm~{\em\ref{algo:sb-mbrb}} is optimal with respect to the pair of
values $\langle \tb,\tm \rangle$.
\end{theorem}

\begin{proof}
Theorem~\ref{theo-necessity} has shown that the condition
$n > 3\tb+2\tm$ is necessary, while Algorithm~{\ref{algo:sb-mbrb}}
has shown that this condition is sufficient
(Theorem~\ref{theo:sb-mbr-correctness}).
\end{proof}


\section{Conclusion}
\label{sec:conclusion}
This article has presented a new communication abstraction
(denoted \MBR) 
that extends Byzantine reliable broadcast (as defined by Bracha
and Toueg~\cite{B87,BT85}) to systems where, at the underlying implementation level, an adversary may suppress
some subset of implementation messages used by the processes to co-operate. From a practical point of view,
this kind of message loss captures 
phenomena such as silent churn, input-disconnection, etc. A  signature-based algorithm implementing the corresponding 
Byzantine-tolerant reliable broadcast in the presence of a message adversary has been
presented and proven correct. This algorithm assumes $n > 3\tb+2\tm$
(where $n$ is the number of processes, \tb is the maximum number of
Byzantine processes, and \tm is an upper bound on the power of the message adversary),
which has been shown to be a necessary and sufficient condition.
message adversary), 

When there is no message adversary, this algorithm is optimal 
    from both Byzantine resilience and the number of
    communication steps. 
  These properties are also satisfied in other circumstances including a message adversary whose power
  $\tm $ is restricted to some well-defined threshold.


\section*{Acknowledgments}
\sloppy This work was partially supported by 
the French ANR projects ByBloS (ANR-20-CE25-0002-01)
and PriCLeSS (ANR-10-LABX-07-81)
devoted to the modular design of building blocks for
large-scale Byzantine-tolerant multi-users applications.
The authors want to thank Colette Johnen, Elad Schiller, 
and Stefan Schmid
for their kind invitation to participate in the SSS 2021 conference. 

\bibliographystyle{plain}

\input{bibliography}
\appendix

\end{document}

%% file: setup.tex
\usepackage{fullpage}
\usepackage[english]{babel}
\usepackage[T1]{fontenc}
\usepackage{times}

\usepackage{amsmath,amssymb,amsthm}
\usepackage{float}
\usepackage{xspace}

\usepackage{tikz}
\usetikzlibrary{arrows.meta}
\usetikzlibrary{decorations.pathreplacing}

\newfloat{algorithm}{thp}{lop}
\floatname{algorithm}{Algorithm}
\newcounter{linecounter}
\renewcommand{\line}[1]{\refstepcounter{linecounter}\label{#1}(\arabic{linecounter})}
\newcommand{\resetline}[1]{\setcounter{linecounter}{0}#1}

\newtheorem{theorem}{Theorem}
\newtheorem{lemma}{Lemma}

\DeclareMathOperator{\Forall}{\forall\,}

\newcommand{\br}{brb\xspace}

\newcommand{\mbr}{mbrb\xspace}
\newcommand{\MBR}{MBRB\xspace}
\newcommand{\MBRB}{MBRB\xspace}

\newcommand{\mbrbassum}{mbrb-Assumption\xspace}

\newcommand{\send}{\ensuremath{\mathsf{send}}\xspace}
\newcommand{\broadcast}{\ensuremath{\mathsf{broadcast}}\xspace}
\newcommand{\received}{\ensuremath{\mathsf{received}}\xspace}

\newcommand{\mbrbroadcast}{\ensuremath{\mathsf{\mbr\_broadcast}}\xspace}
\newcommand{\mbrdeliver}{\ensuremath{\mathsf{\mbr\_deliver}}\xspace}

\newcommand{\brbroadcast}{\ensuremath{\mathsf{\br\_broadcast}}\xspace}

\newcommand{\mmax}{\ensuremath{\mathsf{min}}\xspace}

\newcommand{\tb}{\ensuremath{t}\xspace}
\newcommand{\tm}{\ensuremath{d}\xspace}

\newcommand{\lgd}{\ensuremath{\ell}\xspace}
\newcommand{\rtc}{\ensuremath{\lambda}\xspace}
\newcommand{\omc}{\ensuremath{\mu}\xspace}

\newcommand{\klem}{\ensuremath{k}\xspace}
\newcommand{\llem}{\ensuremath{\ell_2}\xspace}
\newcommand{\llemmin}{\ensuremath{\ell_{2,\mathsf{min}}}\xspace}
\newcommand{\qlem}{\ensuremath{q}\xspace}
\newcommand{\fnlem}{\ensuremath{f}\xspace}

\newcommand{\sAc}{\ensuremath{s_A}\xspace}
\newcommand{\sBc}{\ensuremath{s_B}\xspace}

\newcommand{\sn}{\ensuremath{\mathit{sn}}\xspace}
\newcommand{\sig}{\ensuremath{\mathit{sig}}\xspace}
\newcommand{\sigs}{\ensuremath{\mathit{sigs}}\xspace}

\newcommand{\msgm}{\textsc{msg}\xspace}
\newcommand{\bundlem}{\textsc{bundle}\xspace}

\newcommand{\imp}{imp-message\xspace}
\newcommand{\app}{app-message\xspace}
\newcommand{\imps}{imp-messages\xspace}
\newcommand{\apps}{app-messages\xspace}



%% file: sb-mbrb.tex
\centering
\fbox{
\begin{minipage}[t]{150mm}
\footnotesize
\renewcommand{\baselinestretch}{2.5}
\resetline
\begin{tabbing}
AAA\=aA\=aA\=aA\=aA\=\kill

{\bf operation} $\mbrbroadcast(m,\sn)$ {\bf is}\\

\line{SB-MBRB-save-own-sig-init}
\> save signature for $(m,\sn,i)$ by $p_i$;\\

\line{SB-MBRB-bcast}
\> \broadcast $\bundlem(m,\sn,i,\{\text{all saved signatures for }(m,\sn,i)\})$.\\~\\

{\bf when} $\bundlem(m,\sn,j,\sigs)$ {\bf is} $\received$ {\bf do}\\

\line{SB-MBRB-cond-vld}
\> {\bf if} \big($(-,\sn,j)$ not already \mbr-delivered\\
\>\> $\land$ \sigs contains the valid signature for $(m,\sn,j)$ by $p_j$\big) {\bf then}\\

\line{SB-MBRB-save-sigs}
\>\> save all unsaved valid signatures for $(m,\sn,j)$ of \sigs;\\

\line{SB-MBRB-cond-fwd}
\>\> {\bf if}
\big($(-,\sn,j)$ not already signed by $p_i$\big) {\bf then}\\

\line{SB-MBRB-save-own-sig-fwd}
\>\>\> save signature for $(m,\sn,j)$ by $p_i$;\\

\line{SB-MBRB-fwd}
\>\>\> \broadcast $\bundlem(m,\sn,j,\{\text{all saved signatures for }(m,\sn,j)\})$\\

\line{SB-MBRB-end-cond-fwd}
\>\> {\bf end if};\\

\line{SB-MBRB-cond-dlv}
\>\> {\bf if}
\big(strictly more than $\frac{n+\tb}{2}$ signatures for $(m,\sn,j)$ are saved\big) {\bf then}\\

\line{SB-MBRB-bcast-quorum}
\>\>\> \broadcast $\bundlem(m,\sn,j,\{\text{all saved signatures for }(m,\sn,j)\})$;\\

\line{SB-MBRB-dlv}
\>\>\> $\mbrdeliver(m,\sn,j)$\\

\line{SB-MBRB-end-cond-dlv}
\>\> {\bf end if}\\

\line{SB-MBRB-end-cond-vld}
\> {\bf end if}.
\end{tabbing}
\normalsize
\end{minipage}
}

%% file: msg-dist.tex
\centering
\begin{tikzpicture}
    \def\heightA{2.5}
    \def\heightB{1.5}
    
    \def\widthA{3}
    \def\widthAB{6}
    
    \def\colorA{black}
    \def\colorB{black}
    \def\colorC{black}

    \draw[-{>[scale=2.5,length=2,width=3]}] (0,-.25) -- (0,\heightA+.5);
    \node[align=center] at (-1.5,\heightA) {\# received\\signatures};
    \draw[-{>[scale=2.5,length=2,width=3]}] (-.25,0) -- (\widthAB+.5,0);
    \node[align=center] at (\widthAB+1.5,0) {\# correct\\processes};
    \draw (\widthAB,0) -- (\widthAB,-.1);
    
    \node at (-.4,\heightA) {\klem};
    \draw (-.1,\heightA) -- (\widthA,\heightA) -- (\widthA,-.1);
    \node at (\widthA,-.4) {\llem};
    \draw[\colorA,decorate,decoration={brace}] (\widthA-.1,-.1) -- (.1,-.1);
    \node[\colorA] at (\widthA/2,-.5) {$A$};
    \node[\colorA] at (\widthA/2,\heightA/2) {\sAc};
    
    \node at (-1.15,\heightB) {$\big\lfloor \frac{n+\tb}{2} \big\rfloor = \qlem$};
    \draw (0,\heightB) -- (-.1,\heightB);
    \draw[dashed] (0,\heightB) -- (\widthA,\heightB);
    \draw (\widthA,\heightB) -- (\widthAB,\heightB) -- (\widthAB,-.1);
    \node at (\widthAB,-.4) {$c$};
    \draw[\colorB,decorate,decoration={brace}] (\widthAB-.1,-.1) -- (\widthA+.1,-.1);
    \node[\colorB] at ({(\widthAB+\widthA)/2},-.5) {$B$};
    \node[\colorB] at ({(\widthAB+\widthA)/2},\heightB/2) {\sBc};
\end{tikzpicture}